\documentclass[prl,floatfix,amsmath,superscriptaddress,twocolumn]{revtex4}
\usepackage{amssymb}
\usepackage{graphicx}
\usepackage{graphics}
\usepackage{amsmath}
\usepackage{amsthm}
\usepackage{color}

\def\h{{\bf h}}
\def\g{{\bf g}}

\newcommand{\bea}{\begin{eqnarray}}
\newcommand{\eea}{\end{eqnarray}}
\newcommand{\e}{\eta}

\def\bi{\begin{itemize}}
\def\ei{\end{itemize}}
\def\bc{\begin{center}}
\def\ec{\end{center}}

\def\C{\hbox{$\mit I$\kern-.7em$\mit C$}}
\def\R{\hbox{$\mit I$\kern-.6em$\mit R$}}

\def\ket#1{|#1\rangle}
\newcommand{\one}{\mbox{$1 \hspace{-1.0mm}  {\bf l}$}}
\def\tr{\mathrm{tr}}
\def\ket#1{\left| #1\right>}
\def\bra#1{\left< #1\right|}

\newcommand{\proj}[1]{\ket{#1}\bra{#1}}

\newtheorem{theorem}{Theorem}

\newtheorem{appth}{Theorem}

\begin{document}

\author{J.I. de Vicente}
\affiliation{Departamento de Matem\'aticas, Universidad Carlos III
de Madrid, Legan\'es (Madrid), Spain}
\author{C. Spee}
\affiliation{Institute for Theoretical Physics, University of
Innsbruck, Innsbruck, Austria}
\author{B. Kraus}
\affiliation{Institute for Theoretical Physics, University of
Innsbruck, Innsbruck, Austria}
\title{The maximally entangled set of multipartite quantum states}

\begin{abstract}

Entanglement is a resource in quantum information theory when state
manipulation is restricted to Local Operations assisted by Classical
Communication (LOCC). It is therefore of paramount importance to
decide which LOCC transformations are possible and, particularly,
which states are maximally useful under this restriction. While the
bipartite maximally entangled state is well known (it is the only
state that cannot be obtained from any other and, at the same time,
it can be transformed to any other by LOCC), no such state exists in
the multipartite case. In order to cope with this fact, we introduce
here the notion of the Maximally Entangled Set (MES) of $n$-partite
states. This is the set of states which are maximally useful under
LOCC manipulation, i.\ e.\ any state outside of this set can be
obtained via LOCC from one of the states within the set and no state
in the set can be obtained from any other state via LOCC. We
determine the MES for states of three and four qubits and provide a
simple characterization for them. In both cases, infinitely many
states are required. However, while the MES is of measure zero for
$3$-qubit states, almost all $4$-qubit states are in the MES. This
is because, in contrast to the $3$-qubit case, deterministic LOCC
transformations are almost never possible among fully entangled
four-partite states. We determine the measure-zero subset of the MES
of LOCC convertible states. This is the only relevant class of
states for entanglement manipulation.

\end{abstract}
\maketitle

Multipartite entangled states constitute the essential ingredient
for many fascinating applications within quantum computation and
quantum communication \cite{Gothesis97,RaBr01}. The theory of
many-body states also plays an important role in other fields of
physics \cite{AmFa08}. As the existence of those practical and
abstract applications rests upon the subtle properties of
multipartite entangled states, one of the main goals in quantum
information theory is to gain a better understanding of the
non-local properties of quantum states. Whereas the bipartite case
is well understood, the theory of multipartite entanglement is still
in its infancy \cite{reviews}.

In the context of quantum information theory, entanglement is a
resource that allows one to perform certain information processing
tasks. This has led to the development of entanglement theory, which
deals with the qualification and quantification of entanglement and
with the manipulation of this resource in general \cite{reviews}.
Therein, the notion of Local Operations assisted by Classical
Communication (LOCC) plays a central role as this is the most
general form of manipulating a multipartite state by spatially
separated parties. Thus, LOCC convertibility induces the natural
ordering in the set of entangled states and the very fundamental
condition for a function to be an entanglement measure is that it
does not increase under LOCC. Hence, entanglement theory is a
resource theory in which entanglement is a resource for
manipulations restricted to LOCC. It is then a fundamental question
to ask which states are maximally entangled, i.\ e.\ which states
are maximally useful under the LOCC restriction. Nielsen
\cite{nielsen} has characterized all LOCC transformations which are
possible among arbitrary pure states in the bipartite case. From
there, it follows that the Bell state (and its generalization to
higher dimensions) is \textit{the} maximally entangled state as this
is the only state that can be transformed to any other by LOCC
operations. Thus, it is not surprising that the Bell state plays a
prominent role in most known bipartite quantum information protocols
such as teleportation \cite{teleportation} and cryptography
\cite{ekert} as well as in quantification of entanglement
\cite{DistillBenett}. Unfortunately, there is no straightforward
extension of Nielsen's theorem to the multipartite case and the
investigation of LOCC transformations is very difficult due to their
complicated mathematical structure \cite{chitambar1}. Indeed, LOCC
convertibility is in general unknown in the multipartite case save
for a few classes of states \cite{turgut}. For this among other
reasons, multipartite entanglement has been classified according to
other physically and/or mathematically motivated operations. Local
unitary operations (LUs) are the only invertible LOCC
transformations and, hence, they interconvert states with the same
entanglement. Recently, necessary and sufficient conditions for the
LU equivalence of pure $n$-qubit states have been derived
\cite{barbara}. Stochastic LOCC (SLOCC) operations identify states
which can be interconverted by LOCC non-deterministically but with a
non-zero probability of success. There exist two different SLOCC
classes for $3$-qubit states \cite{slocc3} but there are infinitely
many for more parties \cite{slocc4}. Both LU and SLOCC identify
fundamentally different forms of entanglement but they just define
equivalence classes and cannot be utilized to identify which states
are more useful than others. Thus, another approach has been
investigating separable transformations (SEP) (see \cite{Gour} and
references therein), which although without a clear operational
meaning, have a much simpler mathematical structure and they include
LOCC. However, the inclusion is strict \cite{sep} and there exist
SEP transformations that cannot be implemented by LOCC
\cite{sepnotlocc}. Other authors have tried to identify maximally
entangled multipartite states by extrapolating different particular
properties of the Bell state to the multipartite case \cite{maxent}.
However, let us stress here that, as argued above, LOCC
transformations induce the only operationally meaningful ordering in
the set of entangled states. Consequently, multipartite maximal
entanglement can only be rigorously established on the grounds of
maximal usefulness under LOCC. This is precisely the main aim of
this Letter.

Despite all the difficulties one faces when investigating LOCC
transformations, we introduce here the notion of the Maximally
Entangled Set (MES) of $n$-partite states. A MES $MES_n$ is the
minimal set of $n$-partite states such that any other truly
$n$-partite entangled state can be obtained deterministically from a
state in $MES_n$ via LOCC. In other words, $MES_n$ is a set of
states with the following properties: (i) No state in $MES_n$ can be
obtained from any other $n$-partite state via LOCC (excluding LU)
and (ii) for any $n$-partite state, $\ket{\Phi} \not\in MES_n$,
there exists a state in $MES_n$ from which $\ket{\Phi}$ can be
obtained via LOCC. Thus, it is the set of states which is maximally
entangled. We consider the simplest nontrivial cases of few-qubit
systems and we determine $MES_3$ and $MES_4$ \footnote{In this
paper, due to space restrictions, we only consider generic $4$-qubit
states, i.\ e.\ we only exclude a subset of states of measure
zero.}. Contrary to the bipartite case, the sets do not contain a
single state but infinitely many. Nevertheless, $MES_3$ is of
measure zero in the full set of $3$-qubit states and, hence, very
few states are maximally useful under LOCC. The situation changes
drastically already in the four-partite case as we show that $MES_4$
is of full measure. Thus, perhaps surprisingly, almost all $4$-qubit
states are in the MES. The reason for this is that almost all states
are \textit{isolated}, i.\ e.\ they cannot be obtained from nor
transformed to any other fully entangled state by deterministic LOCC
(excluding LU). Hence, LOCC induces a trivial ordering in the set of
entangled states and the possibility of LOCC conversions is very
rare in the multipartite case. This implies, that most states are
useless for entanglement manipulation via LOCC. However, we also
identify a zero-measure subset of states in the MES which are LOCC
convertible. Those are the most useful states regarding entanglement
manipulation. Hence, the investigation of this class of states
could, as was the case for the Bell state, lead to new multipartite
applications of quantum information.

Troughout the paper we denote the identity operator and the Pauli
operators by $\one, X, Y, Z$ (as well as by $\sigma_i$ with
$i=0,1,2,3$). Moreover, $W(\alpha)=\exp(i\alpha W)$ for $W=X,Y,Z$.
Whenever it does not lead to any confusion we ignore normalization.
$\mathcal{G}$ denotes the set of local invertible (not necessarily
determinant $1$) operators and $g,h$ denote elements of
$\mathcal{G}$, e.g. $g=g_1\otimes \ldots \otimes g_n$, with $g_i\in
GL(2)$. Two states are said to be in the same SLOCC class (LU class)
if there exists a $g\in \mathcal{G}$ ($g$ local unitary) which maps
one state to the other.

Notice that when studying LOCC convertibility, we ignore pure LU
transformations which can always be performed. Hence, in rigor, we
characterize LOCC convertibility among the LU-equivalence classes
and we only consider one representative for each class. Notice as
well that we consider LOCC transformations among fully-entangled
states which are, hence, only possible among states in the same
SLOCC class.

For the present work it is relevant to review the results on the
mathematically more tractable SEP presented in \cite{Gour}. We
denote by $S(\Psi)=\{S\in \mathcal{G}: S\ket{\Psi}=\ket{\Psi}\}$ the
set of symmetries of $\ket{\Psi}$. In \cite{Gour} it has been shown
that a state $\ket{\Psi_1} = g \ket{\Psi}$ can be transformed via
SEP to $\ket{\Psi_2}= h \ket{\Psi}$ iff there exists a $m \in N$ and
a set of probabilities, $\{p_k\}_{1}^m$ ($p_k\geq 0, \sum_{k=1}^m
p_k=1$) and $S_k \in S(\Psi)$ such that \bea \label{EqSep} \sum_k
p_k S_k^\dagger H S_k= r G.\eea Here, $H=h^\dagger h \equiv
\bigotimes H_i$, and $G=g^\dagger g \equiv \bigotimes G_i$ are local
operators and $r=\frac{n_{\Psi_2}}{n_{\Psi_1}}$ with
$n_{\Psi_i}=||\ket{\Psi_i}||^2$. The local POVM elements
accomplishing the task to transform $\ket{\Psi_1}$ into
$\ket{\Psi_2}$ are given by $M_k=\frac{ \sqrt{p_k}}{\sqrt{r}} h S_k
g^{-1}$. Note that if all $S_k$ are unitary then Eq. (\ref{EqSep})
implies that $r G \prec H$. That is, the eigenvalues of the positive
operator $ r G$ are majorized by the eigenvalues of $H$
\cite{nielsen,Gour}.

We first use this result to set necessary conditions for LOCC
convertibility. The basic idea here is to show that any non-unitary
symmetry can be used to define a standard form up to LUs for the
different SLOCC classes. Once this standard form is fixed, the only
operations which make a transformation possible are the unitary
symmetries \footnote{In a future publication we will present a
systematic method to determine the symmetry of a state $\ket{\Psi}$
to analyze all possible LOCC transformations; however, for the
classes considered here, the determination of the corresponding
symmetries is straightforward.}. Later, we show that these
transformations are already so constrained that, whenever possible,
one can also find a corresponding LOCC protocol. Moreover, as we
will see, the LOCC protocols are always quite simple. In fact, each
party needs to measure at most once.

Before we investigate the MES for three and four qubits, it will be
useful to keep in mind that any state with the property such that
all single-qubit reduced states are completely mixed is in the MES
\cite{Bennett99}. Thus, any connected graph state as well as any
error correcting code is in the MES.

%%%%%%%%

Let us now determine $MES_3$, the MES for three qubits. As there
exist two inequivalent tripartite entangled SLOCC classes, the GHZ
class and the W class \cite{slocc3}, $MES_3$ must include at least
two states. Notice that these are the only classes of multipartite
states for which LOCC transformations have been characterized
\cite{turgut}. Since this exhausts all the possible classes
(exclusively) for $3$-qubit states, one could also determine $MES_3$
from the results therein. However, in order to demonstrate our
techniques, we now derive $MES_3$ independently from those results.

Let us begin with the GHZ class. Since the single qubit reduced
states of the GHZ state $\ket{GHZ}= \ket{000}+\ket{111}$ are
completely mixed, we know that it is in $MES_3$. It is well known
that any element of the symmetry group of the GHZ state can be
written as \cite{verstraete} $\bar{X}^k P_{\vec \gamma}$, for $k\in
\{0,1\}$. Here, $P_{\vec{\gamma}}=P_{\gamma_1}\otimes
P_{\gamma_2}\otimes P_{(\gamma_1\gamma_2 )^{-1}}$ with
$P_\gamma=diag(\gamma, \gamma^{-1})$ and $\gamma_i \in \C$ and
$\bar{X}\equiv X^{\otimes 3}$. We now use this symmetry to determine
a standard form for the states in the GHZ class. For any positive
definite $2\times 2$ matrix $g^\dag g$ there exists a $\gamma$ such
that $g^\dag g=P_\gamma^\dagger g_x^\dag g_x P_\gamma $ where
$g_x^\dag g_x\in span\{\one, X\}$ and $\tr(g_x^\dag g_xX)\geq0$.
Thus, any state in the GHZ class can be written (up to LUs
\footnote{The LUs are chosen such that $g_x=\sqrt{g_x^\dagger
g_x}$.}) as \bea \label{GHZsf} g^1_x \otimes g^2_x\otimes (g^3_x
P_{z})\ket{GHZ},\eea with $z\in \C$. Here, and in the following,
$g_w^i \in span\{\one, W\}$, for $W\in \{X,Y,Z\}$ such that
$G_w^i=(g_w^i)^\dagger g_w^i= 1/2\one + g^{(i)}_1 W$, where
$g^{(i)}_1 \in [0,1/2)$ to ensure that the operators are invertible
(otherwise entanglement is destroyed). In order to simplify the
notation, we will allow for negative values of $g^{(3)}_1$ in the
following, even though the corresponding state would be of the form
as in Eq. (\ref{GHZsf}) for a properly chosen value of $z$.

%%%%
Using this standard form, we show now that all states in the MES
(apart from $\ket{GHZ}$) are of the form $g^1_x \otimes g^2_x\otimes
g^3_x \ket{GHZ}$ (i.e. $z=\pm 1$ in Eq. (\ref{GHZsf})) with no
trivial $g^i_x$ (for details see Appendix A). Let $\ket{\Psi_1}= g
\ket{GHZ}$ be an arbitrary initial state and $\ket{\Psi_2}= h
\ket{GHZ}$ an arbitrary final state, where $g=g_x^1\otimes
g_x^2\otimes g_x^3 P_{z_g}$, with $z_g=|z_g|e^{i\alpha_g}$,  and
similarly for $h$. Due to Eq. (\ref{EqSep}), we have that
$\ket{\Psi_1}$ can be transformed to $\ket{\Psi_2}$ via SEP iff
there exist probabilities $p_{k,{\vec \gamma}}$ such that
$\sum_{k,{\vec \gamma}} p_{k,{\vec \gamma}} \bar{X}^k P_{\vec
\gamma}^\dagger H P_{\vec \gamma}\bar{X}^k=r G$. Equating the
$\ket{lll}\bra{mmm}$ matrix elements for $l,m=0,1$, it follows that
if none of the $h_x^i$ is trivial, then a state with $z_h=\pm 1$ can
only be obtained from a state with $z_g=\pm 1$.

On the other hand, any state with $z_h \neq \pm 1 $ or at least one
trivial $h_x^i$ can be obtained from a state with $z_g=\pm 1$. In
fact, for those states one can not only derive a SEP, but a LOCC
protocol which accomplishes this task. For instance, the final state
$\ket{\Psi_2} = h_x^1\otimes h_x^2\otimes h_x^3 P_{z_h} \ket{\Psi}$
(with $z_h \neq \pm 1 $) can be reached from the state $\ket{\Psi_1}
= h_x^1\otimes h_x^2\otimes g_x^3 \ket{\Psi}$ with a properly chosen
operator $g_x^3$ via the following LOCC protocol: Party 3 applies
the POVM $\{M_1\propto h_x^3 P_{z_h} (g_x^3)^{-1},M_2\propto h_x^3
P_{z_h} X (g_x^3)^{-1}\}$ and in case outcome $1$ ($2$) is obtained,
all other parties do nothing (apply a $X$ operation). Since $X$
commutes with $h_x^i$, it can be easily seen that the desired state
is obtained for both outcomes. Thus, the only GHZ-class states that
are in the MES  are of the form $g_x^1\otimes g_x^2\otimes g_x^3
\ket{\Psi}$.

%%%%%%%%%%

Let us now treat the $W$ class. Using the symmetry of the $W$ state
$\ket{W}= \ket{001}+\ket{010}+\ket{100}$ \cite{verstraete} (see also
Appendix A) it is easy to see that any state in the $W$ class can be
written as $g_1\otimes g_2 \otimes \one \ket{W}$, where \bea
g_1=\begin{pmatrix} 1&0\\ 0& x_1(g)/x_3(g)
\end{pmatrix},\quad
  g_2=\begin{pmatrix} x_3(g)&x_0(g)\\ 0& x_2(g) \end{pmatrix}.\eea
We now show that states of this form are in the MES iff $x_0(g)=0$
(for details see Appendix A).

First of all, note that we only need to consider unitary symmetries
since any POVM element transforming the state
$\ket{\Psi_1}=g_1\otimes g_2 \otimes \one \ket{W}$ into
$\ket{\Psi_2}=h_1\otimes h_2 \otimes \one \ket{W}$, would be of the
form $h S_k g^{-1}=h_1 s_k^1 g_1^{-1}\otimes h_2 s_k^2
g_2^{-1}\otimes s_k^3$, which can only be implemented via LOCC if
party 3 just applies a unitary, i.e. if $s_k^3$ is unitary. It can
then be shown that in this case SEP is only possible if the whole
symmetry is unitary.

The only unitaries leaving the $W$ state invariant (up to a global
phase) are of the form $Z(\alpha)^{\otimes 3}$ ($\alpha \in \R$).
Suppose now that $x_0(h)=0$. In this case, $H_2$ as well as $H_1$
commutes with any symmetry operator. Hence, $\ket{\Psi_1}$ can be
transformed into $\ket{\Psi_2}$ via SEP (see Eq. (\ref{EqSep})) iff
$ H_1\otimes H_2 \otimes \one =r G_1\otimes G_2 \otimes \one$.
Clearly, this implies that $x_0(g)=0$. Thus, states with $x_0(h)=0$
can only be obtained from states with $x_0(g)=0$. Moreover, in this
case, the states can only be converted into each other if they are
LU equivalent, since there is only one POVM element. This shows that
states with $x_0(g)=0$ are in the MES. Similar to the GHZ case, one
can construct a LOCC protocol which reaches any state with
$x_0(h)\neq 0$ from one with $x_0(g)=0$. This shows the following
theorem.
%%%%%%%%%

\begin{theorem} The MES of three qubits, $MES_3$, is given by \bea \label{mes3} MES_3=\{g_x^1\otimes g_x^2\otimes
g_x^3 \ket{GHZ}, g_1 \otimes g_2 \otimes \one \ket{W}\},\eea where
no $g^i_x\propto \one$ (except for the GHZ state) and $g_1$ and
$g_2$ are diagonal.
\end{theorem}

Interestingly, $MES_3$ has a very simple parametrization in terms of
the decomposition of $3$-qubit states that was introduced by some of
us in \cite{us}. While $3$-qubit LU classes are parameterized by 5
real parameters, it can be shown that any state in $MES_3$ belongs
(up to LUs) to the three-parameter set (see Appendix A) \bea
\{\ket{\Psi(a,\beta,\beta')}\equiv
\ket{0}\ket{\Psi_s}+\ket{1}Y(\beta')\otimes
Y(\beta)\ket{\Psi_s}\},\eea where
$\ket{\Psi_s}=a\ket{00}+\sqrt{1-a^2}\ket{11}$ is in Schmidt
decomposition and $a,\beta, \beta^\prime \in \R$. Using this form,
we show now that any state in $MES_3$ can be mapped into some other
state (outside of $MES_3$) via a non-trivial LOCC protocol, which
also implies that no $3$-qubit state is isolated. Note that
$X\otimes ZY(-\beta')\otimes ZY(-\beta)$ leaves any state
$\ket{\Psi(a,\beta,\beta')}$ invariant. Using this symmetry, it is
easy to see that any $\ket{\Psi(a,\beta,\beta')}$ can be transformed
into any state $A\otimes \one^{\otimes 2}
\ket{\Psi(a,\beta,\beta')}$, where $A$ is such that $\tr(A^\dagger A
X)=0$ and  $\tr(A^\dagger A)=1$. The corresponding POVM is locally
realizable and it is given by POVM elements $M_1=A \otimes \one,
M_2= A X  \otimes ZY(-\beta')\otimes ZY(-\beta)$.

In summary, although it contains infinitely many states, the MES for
three qubits is of measure zero. Furthermore, no 3-qubit state is
isolated, i.\ e.\ LOCC entanglement manipulation from every state in
the MES is always possible.

We move now to the $4$-qubit case. Since there are infinitely many
SLOCC classes, $MES_4$ must contain infinitely many states. Generic
states belong to the different SLOCC classes known as $G_{abcd}$
with the representatives \cite{slocc4}
\begin{align}
|\Psi\rangle&=\frac{a+d}{2}(|0000\rangle+|1111\rangle)+\frac{a-d}{2}(|0011\rangle+|1100\rangle)\nonumber\\
&+\frac{b+c}{2}(|0101\rangle+|1010\rangle)+\frac{b-c}{2}(|0110\rangle+|1001\rangle),\label{seed}
\end{align}
where $a,b,c,d \in \C$ with $a\neq\pm b$ etc.. In the following, the
states of the form in Eq. (\ref{seed}) will be called the seed
states. Due to the normalization and the irrelevant global phase,
the seed states are parameterized by $6$ parameters. It can be
easily seen that the symmetry in this case is given by
$\{\sigma_i^{\otimes4}\}_{i=0}^4$. For any state $\ket{\Phi}$ in any
of these SLOCC classes, there exists a local invertible matrix $g\in
\mathcal{G}$ and a seed state $\ket{\Psi}$ such that $\ket{\Phi}= g
\ket{\Psi}$. Without loss of generality, we normalize the positive
operators $G_i=g_i^\dagger g_i$ such that $\tr(G_i)=1$ and use the
notation $G_i=1/2\one + \sum_k g_k^{(i)} \sigma_k$, with
$g_k^{(i)}\in \R$. Since we consider fully entangled $4$-qubit
states, $G_i$ is a positive full-rank operator, i.e.  $0\leq
|\g^{(i)}| < 1/2$, where $\g^{(i)}=(g^{(i)}_1,g^{(i)}_2,g^{(i)}_3)$.
Considering the trace of Eq. (\ref{EqSep}), we obtain $r=1$. Note
that, for each $i$, the vector containing the eigenvalues of $G_i$,
which are $1/2 \pm |\g^{(i)}|$, must be majorized by the
corresponding vector for $H_i$. Thus, $|\g^{(i)}|$ cannot decrease
under LOCC, and therefore these parameters behave monotonically
under LOCC. The symmetry of the seed states only allows to
simultaneously change for all $i$ the sign of two parameters
($g_k^{(i)}$ and $g_l^{(i)}$, where $k,l \in \{1,2,3\}$ and $k \neq
l$) in $G_i$. Thus, the matrices $G_i$ can be made unique. Moreover,
by sorting the coefficients in Eq. (\ref{seed}) this leads to a
unique standard form. This implies that two generic states are LU
equivalent iff their standard forms coincide. Since four-qubit LU
equivalence classes are parameterized by 18 parameters \cite{LU4},
and the set of states considered here is parameterized by the $6$
independent seed parameters and the $12$ independent SLOCC
parameters, $g^{(i)}_j$, for $i\in\{1,2,3,4\}, j\in\{1,2,3\}$, the
set is, as expected, of full measure (and dense) in the set of
four-qubit states.

We now study  $MES_4$, the MES of four qubits. Notice that all
single qubit reduced states of a seed state are completely mixed,
which implies that all seed states are in $MES_4$. We proceed now as
in the three-partite case. First, we derive necessary conditions for
a state to be reachable via SEP, then we derive the corresponding
LOCC protocol (if it exists). Contrary to before, we will see that
almost no states can be reached via LOCC.

%%%%%%%%%%%%

One of the main differences between the four-partite and the
three-partite case is that there are only finitely many symmetries
and that all symmetries are unitary. This fact can be used to derive
very simple necessary conditions for the existence of the LOCC
protocol. The main idea here is to observe that for
$\ket{\Psi_1}=g\ket{\Psi}$ and $\ket{\Psi_2}=h\ket{\Psi}$, Eq.
(\ref{EqSep}) implies that

\bea \label{EQ_Tensor}{\cal E}_{4}(H)={\cal E}(H_1)\otimes {\cal
E}(H_2)\otimes {\cal E}(H_3)\otimes {\cal E}(H_4), \eea where ${\cal
E}_4$ is the completely positive map given in Eq.\ (\ref{EqSep}), i.
e. ${\cal E}_4(\rho)=\sum_{k=0}^3 p_k \sigma_k^{\otimes 4} \rho
\sigma_k^{\otimes 4}$,  and ${\cal E}(\rho)=\sum_{k=0}^3 p_k
\sigma_k \rho \sigma_k$. Note that Eq.\ (\ref{EQ_Tensor}) only
depends on the state $\ket{\Psi_2}$ and is independent of
$\ket{\Psi_1}$. Considering Eq.\ (\ref{EQ_Tensor}) for two systems,
i.\ e.\ tracing over the other two, already yields very strong
necessary conditions for $\ket{\Psi_2}$ to be reachable via SEP. In
\cite{supmat}, we show which of those states can indeed be reached
via LOCC by constructing, analogous to the three-partite case, the
corresponding LOCC protocol. With that, we obtain the following
theorem (for details of the proof see Appendix B).

%%%%%%%%%%%%%%%

\begin{theorem} \label{Thgen4} A generic state, $h \ket{\Psi}$,
is reachable via LOCC from some other state iff (up to permutations)
 either $h= h^1\otimes h_w^2 \otimes h_w^3 \otimes h_w^4$, for
$w\in\{x,y,z\}$ where $h^1\neq h_w^1$ or $h= h^1\otimes\one^{\otimes
3}$ with $h^1\not \propto \one$ arbitrary.
\end{theorem}
The states in Theorem \ref{Thgen4} form a family defined by only 12
parameters, which implies that the set of states that can be reached
via LOCC is of measure zero. Therefore, all the remaining generic
states are necessarily in $MES_4$, which is then of full measure and
contains almost all states.

Let us now characterize which states can indeed be used for
entanglement manipulation, i.e. which of them can be transformed by
LOCC into another state. Since a state $\ket{\Psi_1}=g\ket{\Psi}$
having this property can only be transformed into some state
$\ket{\Psi_2}=h\ket{\Psi}$ given in Theorem \ref{Thgen4}, the
conditions ${\cal E}(H_i)=G_i$ [see Eq.\ (\ref{EqSep})] imply that
$G_i=(g_w^i)^\dagger g_w^i$ for $i=2,3,4$. Indeed, one can easily
show that any state obeying the conditions above allows for
non-trivial LOCC transformations (see Appendix B). Thus, we have the
following theorem, which shows that deterministic LOCC manipulations
among fully entangled $4$-qubit states are almost never possible.

\begin{theorem} \label{theoremISO}
A generic state $g \ket{\Psi}$ is convertible via LOCC to some other
state iff (up to permutations) $g=g^1\otimes g_w^2 \otimes g_w^3
\otimes g_w^4$ with $w\in \{x,y,z\}$ and $g^1$ arbitrary.
\end{theorem}

Combining Theorem \ref{Thgen4} and Theorem \ref{theoremISO}, we see
that every state that can be reached via LOCC can at the same time
be transformed into another state, and that all states that are not
of the form given in Theorem \ref{theoremISO} (which are almost all)
are isolated. Moreover, the non-isolated generic states which are in
the MES constitute a $10$-parameter family of the form
$G_w\ket{\Psi}$, with $G_w=g_w^1 \otimes g_w^2 \otimes g_w^3 \otimes
g_w^4$ and $w\in \{x,y,z\}$ (excluding the case where $g_w^i\not
\propto \one$ for exactly one $i$). Thus, the set of generic
$4$-qubit entangled states is divided into two subsets with very
different physical properties: a measure-zero set of
LOCC-convertible states
 and a full-measure
set of isolated states. Clearly, the first set appears to be the
physically more relevant one. In particular, its intersection with
the MES, i.e. the set $\{g_w^1\otimes g_w^2 \otimes g_w^3 \otimes
g_w^4 \ket{\Psi}\}$ with $w\in \{x,y,z\}$, gives rise to the most
useful states under LOCC manipulation. Note that all of these states
can be written in a very simple form as
$\ket{0}\ket{\Psi_0}+\ket{1}X^{\otimes 3} \ket{\Psi_0}$, with
$\ket{\Psi_0}$ depending on the SLOCC parameters $\{g_j^{(i)}\}$ and
the seed parameters.

In summary, we have introduced the concept of the MES of $n$-qubit
states $MES_n$, the analogue of the maximally entangled state in the
bipartite case. We have explicitly derived $MES_3$ and $MES_4$ (for
generic states in the latter case) and have shown that almost all
four-qubit states are in $MES_4$. For more than two parties, the MES
contains infinitely many states; however, while $MES_3$ is of
measure zero (and no state is isolated), $MES_4$ is of full measure
because almost all states are isolated. These results imply that
almost all entangled $4$-qubit states are incomparable according to
the LOCC paradigm inducing a rather trivial ordering in the set of
entangled states and implying that almost all states are useless for
entanglement manipulation. However, we determined the zero-measure
subset of generic non-isolated states in the MES, which allows for a
very simple decomposition. A more in-depth understanding of the
physical and mathematical properties of this class of states might
lead to new insights in multipartite entanglement and its
applications. Also, our results open the doors to investigate
well-studied phenomena in the bipartite case which remained mostly
unexplored so far in the multipartite setting. These include the
identification of new entanglement measures, entanglement catalysis
\cite{catalysis}, optimal probabilities of LOCC conversions
\cite{vidal1} or possibly transformations to mixed states
\cite{vidal2} to name a few. It will also be interesting to study
the MES in the case of LOCC transformations among multiple copies of
states. In a forthcoming article \cite{MESlong}, we investigate all
possible state transformations within the four qubit case. We
anticipate that, even in the non-generic case, very few states in
the MES allow for LOCC conversions. Our tools can also be used to
decide convertibility among states of more than four parties (or
even higher dimensions) as, among other reasons, the validity of
Eq.\ (\ref{EqSep}) is independent of the number and dimensions of
the subsystems.

This research was funded by the Austrian Science Fund (FWF):
Y535-N16.

\section{Appendix A: 3-qubit MES}
We provide here the details of the proof of Theorem 1, which we restate here,
\noindent \\ \\
\textit{ {\bf Theorem 1.} The MES of three qubits, $MES_3$, is given by \bea \label{mes3} MES_3=\{g_x^1\otimes g_x^2\otimes
g_x^3 \ket{GHZ}, g_1 \otimes g_2 \otimes \one \ket{W}\},\eea where,
no $g^i_x\propto \one$ (except for the GHZ state)  and $g_1$ and $g_2$ are diagonal.}
\\

The outline of the proof is as follows. Within each SLOCC class, we show first that none of the states in $MES_3$ can be reached via separable maps (SEP). Since SEP includes LOCC, this implies that these states cannot be obtained by any LOCC protocol either. We then prove that all other states can be obtained via LOCC from states that are in the MES. In particular, we present the corresponding LOCC protocols.\\
In order to improve the readability of the appendix, we repeat the results about possible transformation via SEP presented in \cite{appGour}. As in the main text, let us denote by $S(\Psi)=\{S\in \mathcal{G}: S\ket{\Psi}=\ket{\Psi}\}$ the set of symmetries of $\ket{\Psi}$. Then, a state $\ket{\Psi_1} = g \ket{\Psi}$ can be transformed via SEP to $\ket{\Psi_2}= h \ket{\Psi}$ iff
there exists a $m \in N$ and a set of probabilities, $\{p_k\}_{1}^m$ ($p_k\geq 0, \sum_{k=1}^m p_k=1$) and $S_k \in S(\Psi)$ such that \bea \label{EqSep2} \sum_k p_k S_k^\dagger H S_k= r G.\eea
Here, as in the main text, we use the notation $H=h^\dagger h \equiv \bigotimes H_i$ and $G=g^\dagger g \equiv \bigotimes G_i$ corresponding to local operators and $r=\frac{n_{\Psi_2}}{n_{\Psi_1}}$ with $n_{\Psi_i}=||\ket{\Psi_i}||^2$.\\
In order to prove Theorem 1, we consider the two three-partite SLOCC classes, the GHZ class and the W class, seperately, in Lemma 1 and 2 respectively.
Let us start with the GHZ class. In the following, we will use the standard form for states in the GHZ class that was introduced in the main text, i. e. \bea g^1_x \otimes g^2_x\otimes (g^3_x P_{z})\ket{GHZ},\eea
with $P_z=diag(z, z^{-1})$ and $z\in \C$. As mentioned before, every element of the symmetry group can be written
as $\bar{X}^k P_{\bf \gamma}$, with $k\in \{0,1\}$. This follows from the fact that the symmetry group of the GHZ state is generated by \cite{appverstraete} $P_{\gamma_1}\otimes P_{\gamma_2}\otimes
P_{(\gamma_1\gamma_2 )^{-1}}\equiv P_{\vec{\gamma}}$, with $\gamma_i
\in \C$ and $\bar{X}\equiv X^{\otimes 3}$ and the fact that $P_\gamma X= X
P_{\gamma^{-1}}$.
\noindent \\ \\
\textit{ {\bf Lemma 1.}  The subset of states in $MES_3$ which are in the GHZ class is given by \bea \label{mesGHZ}\{g_x^1\otimes g_x^2\otimes
g_x^3 \ket{GHZ}, \ket{GHZ}\},\eea where no $g^i_x\propto \one$.}
\begin{proof}
We denote by $\ket{\Psi_1}= g \ket{GHZ}$ an arbitrary initial state and by $\ket{\Psi_2}= h \ket{GHZ}$ an arbitrary final state. Here, $g=g_x^1\otimes g_x^2\otimes g_x^3 P_{z_g}$, with $z_g=|z_g|e^{i\alpha_g}$, and similarly for $h$. Note that $z_g=\pm 1$ implies that $\ket{\Psi_1}$ is an element of the set given in Eq. (\ref{mesGHZ}). Using Eq. (\ref{EqSep2}), we have that a state $\ket{\Psi_1}$ can be transformed into $\ket{\Psi_2}$ via SEP
 iff there exist finitely many probabilities $p_{k,{\vec \gamma}}$ such that $\sum_{k,{\vec \gamma}} p_{k,{\vec \gamma}} \bar{X}^k P_{\vec \gamma}^\dagger H P_{\vec \gamma}\bar{X}^k=r G$.  Considering now the necessary conditions $\tr( \sum_{k,{\vec \gamma}} p_{k,{\vec \gamma}} \bar{X}^k P_{\vec \gamma}^\dagger H P_{\vec \gamma}\bar{X}^k \proj{lll})= r\tr( G \proj{lll})$, for $l=0,1$, we have \bea\nonumber
\sum_{k,{\vec \gamma}} p_{k,{\vec \gamma}} |z_h|^{2(-1)^k}=r |z_g|^2 \mbox{ and } \sum_{k,{\vec \gamma}} p_{k,{\vec \gamma}} |z_h|^{-2(-1)^k}=r |z_g|^{-2}.\eea
We now show that states $\ket{\Psi_2}$ of the form in Eq. (\ref{mesGHZ}) cannot be obtained from any other state via LOCC.
If $|z_h|=1$, the conditions above imply that $|z_g|=1$ and $r=1$, i.e.
$n_{\Psi_2}=n_{\Psi_1}$. The last condition is equivalent to $h_1^{(1)} h_1^{(2)} h_1^{(3)}  \cos(2\alpha_h)=g_1^{(1)} g_1^{(2)} g_1^{(3)} \cos(2\alpha_g)$. Equating the $\ket{000}\bra{111}$ matrix elements of both sides of Eq. (\ref{EqSep2}), we obtain $\sum_{k,{\vec \gamma}} p_{k,{\vec \gamma}} h_1^{(1)} h_1^{(2)} h_1^{(3)} e^{(-1)^k 2 i \alpha_h}= g_1^{(1)} g_1^{(2)} g_1^{(3)} e^{2 i \alpha_g}$. Hence, if $h_1^{(1)} h_1^{(2)} h_1^{(3)} \neq 0$, the condition $\alpha_h\in\{0,\pi\}$ implies that the same must be true for $\alpha_g$. We used here that a state with $\alpha_g=\pm \pi/2$ is LU equivalent to a state with $\alpha_g\in \{0,\pi\}$. This shows that if $h_1^{(1)} h_1^{(2)} h_1^{(3)}\neq 0$, a state with $z_h=\pm 1$ can only be obtained from a state with $z_g=\pm 1$, which are precisely the states given in Eq. (\ref{mesGHZ}).\\
We will prove now that a state $\ket{\Psi_2}=h_x^1\otimes h_x^2
\otimes h_x^3 \ket{GHZ}$, with $h_x^{(i)}\not \propto \one$ $\forall
i$, i. e. $h_1^{(1)} h_1^{(2)} h_1^{(3)} \neq 0$ and $z_h=\pm 1$,
can neither be obtained from any other state of this form nor from
the GHZ state. Note that $\bar{X}$ commutes with $H$. Thus, we only
have to consider the symmetries $P_{{\vec\gamma}}$. Due to the
discussion above, we have that $r=1$ (see Eq.(\ref{EqSep2})). Taking the trace of Eq. (\ref{EqSep2}) leads to
\begin{eqnarray} \nonumber&\sum_{k,{\vec \gamma}}p_{k,{\vec
\gamma}}(|\gamma_1|^2+|\gamma_1|^{-2})(|\gamma_2|^2+|\gamma_2|^{-2})\times\\\nonumber&(|\gamma_1
\gamma_2|^2+|\gamma_1 \gamma_2|^{-2})=8,\end{eqnarray} which is
satisfied iff $|\gamma_i|=1$ for all $\gamma_i$ occurring in the
sum. Thus, we only have to consider unitary symmetries in Eq.
(\ref{EqSep2}). This implies that the vector containing the
eigenvalues of $G_i$ must be majorized by the corresponding vector
of $H_i$, which implies that $|g^{(i)}_1|$ cannot decrease under
LOCC. Using in addition the fact that the norms of the two states coincide,
i.e. $g_1^{(1)}g_1^{(2)}g_1^{(3)}=h_1^{(1)}h_1^{(2)}h_1^{(3)}$
($r=1$) we have that $h_1^{(i)}=g_1^{(i)}$, $\forall i$. Thus,
$\ket{\Psi_1}$ can only be transformed into $\ket{\Psi_2}$ if the
states were LU equivalent.
As already mentioned in the main text, the GHZ state is in $MES_3$, since all the single qubit reduced states of the GHZ state are completely mixed.\\
We will prove now that a state is in the MES only if it is of the form given in Eq. (\ref{mesGHZ}), where no $g^i_x\propto \one$. In order to show
that, we present the explicit LOCC protocols that yield
all the other states in the GHZ class from a state in the MES.
Let us start by showing that states of the form $\ket{\Psi_2}=\one\otimes h_x^2 \otimes h_x^3 \ket{GHZ}$, i. e. $h_1^{(1)} h_1^{(2)} h_1^{(3)} = 0$, can be obtained from the GHZ state via the following LOCC protocol. Party $2$ applies the POVM $\{h_x^2, h_x^2 Z\}$ and party $3$ applies the POVM $\{ h_x^3,  h_x^3 Z\}$. They communicate their outcomes, which we denote by $i_2,i_3\in \{0,1\}$, to party 1. Then, party $1$ applies $Z^{i_2+i_3}$. Due to the symmetry of the GHZ state, it can be easily seen that this LOCC protocol accomplishes the task.\\
The remaining class of states, namely those where $z_h \neq \pm 1 $, can be obtained from a state with $z_g=\pm 1$, as can be seen as follows. We show that for an
arbitrary final state,  $\ket{\Psi_2} = h_x^1\otimes h_x^2\otimes
h_x^3 P_{z_h} \ket{GHZ}$ (with $z_h \neq \pm 1 $), there exists an
operator $g_x^3$ such that the state $\ket{\Psi_1} = h_x^1\otimes
h_x^2\otimes g_x^3 \ket{GHZ}$ can be transformed into $\ket{\Psi_2}$
via LOCC. Consider the POVM $\{M_1,M_2\}$ with \bea M_1=\sqrt{p}
\one^{\otimes 2}\otimes h_x^3 P_{z_h} (g_x^3)^{-1}\eea and \bea
M_2=\sqrt{p} X^{\otimes 2} \otimes h_x^3 P_{z_h} X (g_x^3)^{-1},\eea
with $p=1/(|z_h|^2+1/|z_h|^2)$. It is easy to see that $g_x^3$ can be chosen s. t. $\{M_1,M_2\}$ constitutes a POVM. In particular, $M_1^\dagger M_1+M_2^\dagger M_2=\one$ iff
\bea P_{z_h}^\dagger H_x^3 P_{z_h}+X P_{z_h}^\dagger H_x^3 P_{z_h} X
= G_x^3.\eea The left hand side of this equation is equal to $1/2\one
+\tilde{b} X$, with $\tilde{b}=2 p b \cos(2\alpha_h)$. Thus, for any
$z,b$ one can chose $G_x^3 = 1/2 \one +\tilde{b} X$ to satisfy
the above condition. The LOCC protocol to transform $\ket{\Psi_2}$ to $\ket{\Psi_1}$ is given by the following procedure. Party $3$ applies the POVM \bea\{h_x^3 P_{z_h} (g_x^3)^{-1},h_x^3 P_{z_h} X (g_x^3)^{-1}\}\eea and communicates the outcome to the other parties. In case of outcome $1$, all other parties do nothing, whereas in case of outcome $2$ they apply a $X$ operation. Due to the fact that $X$ commutes with $h_x^i$ and $X^{\otimes 3}$ is an element of the symmetry group of the GHZ state, the desired state is obtained for both outcomes.\\
\end{proof}

The other SLOCC class for truly entangled three-partite states is the $W$ class. The symmetries of the $W$ state \cite{appverstraete}, $\ket{W}= \ket{001}+\ket{010}+\ket{100}$, are given by

\begin{eqnarray}\label{symW}
&S_{x,y,z}=s^1_{x,y}\otimes s^2_{x,z}\otimes s^3_{x,y,z}\\ \nonumber
&\equiv \left(
    \begin{array}{cc}
      x & y \\
      0 & 1/x \\
    \end{array}
  \right)\otimes\left(
    \begin{array}{cc}
      x & z \\
      0 & 1/x \\
    \end{array}
  \right)\otimes\left(
    \begin{array}{cc}
      x & -y-z \\
      0 & 1/x \\
    \end{array}
  \right),
\end{eqnarray}
where $x,y,z\in\mathbb{C}$ and any state in this SLOCC class can be written as
$x_0\ket{000}+x_1\ket{100}+x_2\ket{010}+x_3\ket{001},$ with $x_i \geq 0$.
Another way of representing an arbitrary state in the $W$ class, which we will use in the following, is $g_1\otimes g_2 \otimes \one \ket{W}$, where \footnote{Note that we can choose here wlog $x_3\neq 0$, since we are considering fully entangled states.} \bea g_1=\begin{pmatrix} 1&0\\ 0& x_1(g)/x_3(g) \end{pmatrix}
  g_2=\begin{pmatrix} x_3(g)&x_0(g)\\ 0& x_2(g) \end{pmatrix}.\eea
\noindent \\ \\
\textit{ {\bf Lemma 2.}  The subset of states in $MES_3$ which are in the W class is given by \bea \label{W}\{ g_1 \otimes g_2 \otimes \one \ket{W}\},\eea where $g_1$ and $g_2$ are diagonal.}

\begin{proof}

In order to prove that the states in this class are in the MES if
$x_0(g)=0$, let us first show that we only need to consider unitary
symmetries. Recall from the main text that any element of a POVM
that can be implemented must be such that the operation applied by
party $3$, $s_{x,y,z}^3$, is unitary. This implies that the corresponding
symmetries obey $z=-y$ and $|x|=1$ (see Eq. (\ref{symW})). Inserting
these symmetries, as well as the expressions for $G$ and $H$, where
$x_0(h)=0$, in Eq. (\ref{EqSep2}) and tracing over the third
party results in

\bea\nonumber &\sum_{x,y} p_{x,y} \left(
    \begin{array}{cc}
     1 & x^\ast y \\
      y^\ast x  & |\frac{x_1(h)}{x_3(h)}|^2+|y|^2 \\
    \end{array}
  \right)
\otimes\left(
    \begin{array}{cc}
      |x_3(h)|^2 & b \\
      b^\ast & c\\
    \end{array}
  \right)
\\\nonumber& =r  \left(\begin{array}{cc}
     1 &0\\
    0  & |\frac{x_1(g)}{x_3(g)}|^2 \\
    \end{array}
  \right)
\otimes\left(
    \begin{array}{cc}
      |x_3(g)|^2 & x_0(g) x_3^\ast(g) \\
      x_0^\ast(g) x_3(g) & |x_2(g)|^2+ |x_0(g)|^2\\
    \end{array}
  \right),\eea where $b= -x^\ast y |x_3(h)|^2$ and $c=|y|^2 |x_3(h)|^2+|x_2(h)|^2$. Equating the $\ket{0}\bra{1} \otimes \ket{1}\bra{0}$ matrix elements of both sides of the equation leads to $\sum_{x,y} p_{x,y} |y|^2=0$. Thus, $y$ must be zero and therefore the symmetry is unitary (see Eq. (\ref{symW})).\\
The unitaries that leave the $W$ state invariant (up to a global phase) are given by $Z(\alpha)^{\otimes 3}$, with $\alpha \in \R$. Thus, $\ket{\Psi_1}$ can be transformed into $\ket{\Psi_2}$ via SEP iff there exist finitely many probabilities $p_\alpha$ such that \bea \sum_\alpha p_\alpha H_1 \otimes Z(-\alpha) H_2 Z(\alpha)\otimes \one =r G_1\otimes G_2 \otimes \one,\eea where we used that $H_1$ commutes with any $Z(\alpha)$. Note that if $x_0(h)=0$, then $H_2$ also commutes with $Z(\alpha)$ and therefore, in this case, we have $ H_1\otimes H_2 \otimes \one =r G_1\otimes G_2 \otimes \one$, which implies the condition $x_0(g)=0$. Hence, states with $x_0(h)=0$ can only be obtained via LOCC from states with $x_0(g)=0$. Moreover, since there is only one POVM element, in this case, the states can only be transformed into each other if they are LU equivalent. From this it follows that states with $x_0(g)=0$ are in the MES. \\
Let us now show that these are the only states of the $W$ class in the MES. In order to do so, we show that any state $\ket{\Psi_2}=h_1\otimes h_2 \otimes \one \ket{W}$ with $x_0(h)\neq 0$ can be obtained from a state $\ket{\Psi_1}=h_1\otimes g_2 \otimes \one \ket{W}$ with $x_0(g)=0$. One can choose $g_2$ such that the POVM \bea\{M_1=\one \otimes h_2 g_2^{-1} \otimes \one , M_2=Z \otimes h_2 Z g_2^{-1}\otimes Z\}\eea accomplishes this task. This can be seen as follows. The condition $M_1^\dagger M_1+M_2^\dagger M_2=\one$ is equivalent to $G_2=H_2+ZH_2 Z$, where the right hand side of this equation is equal to \bea 2  \left(\begin{array}{cc}
     x_3(h)^2 &0\\
    0  & x_0(h)^2+x_2(h)^2 \\
    \end{array}
  \right).\eea Thus, choosing $G_2$ as in the equation above (with $x_0(g)=0$) ensures that $\{M_1,M_2\}$ is a POVM. Hence, the LOCC protocol transforming $\ket{\Psi_1}$ into $\ket{\Psi_2}$ is as follows. Party $2$ applies the POVM $\{h_2 g_2^{-1}, h_2 Z g_2^{-1}\}$ and communicates the outcome of the measurement to the other parties. In case outcome $1$ ($2$) is obtained, all remaining parties apply $\one$ ($Z$) respectively to obtain the desired state.
Thus, the states in the MES in this class are of the form
\bea \ket{\Psi_W}=g_1\otimes g_2 \otimes \one \ket{W},\eea where $g_1$ and $g_2$ are diagonal, i.e.  $x_0(g)=0$.
\end{proof}

In summary, in Lemmas 1 and 2  we have characterized the MES for each of the three-partite entangled SLOCC classes seperately. Combining these results directly proves Theorem 1. Note that the non-unitary symmetry was only used to derive a standard form within the different SLOCC classes, whereas the unitary symmetries, i. e. the Pauli operators, lead to the possible LOCC transformations (see Eq. (\ref{EqSep})).

Writing the states that are in the MES in the computational basis, it is easy to show that any state in $MES_3$ is an element of the set
\bea\nonumber \tilde{S}_3=\{\ket{\Psi(a,\beta,\beta')}\equiv \ket{0}\ket{\Psi_s}+\ket{1}Y(\beta')\otimes Y(\beta)\ket{\Psi_s}\},\eea where
$\ket{\Psi_s}=a\ket{00}+\sqrt{1-a^2}\ket{11}$ is in Schmidt decomposition and the parameters, $a$, $\beta$ and $\beta'$ are simple functions of the parameters $\{g_1^{(i)}\}_{i=1}^3$ for the GHZ class and $\{x_i\}_{i=1}^3$ for the $W$ class.\\

Note that $\tilde{S}_3$ also includes biseparable states, e.g. $a=1$, $\beta=0$ and $\beta^\prime$ arbitrary. Thus, $\tilde{S}_3$ is strictly larger than $MES_3$. However, $MES_3$ can be obtained from $\tilde{S}_3$ by excluding state like those in the GHZ class, with $g^{(i)}_x=0$ for some $i$ (excluding the GHZ state), and states which belong to the same LU-equivalence class as one element within the set.

\section{Appendix B: Generic $4$-qubit MES}

In this section, we consider LOCC transformations among generic
$4$-qubit states and prove Theorems 2 and 3 in the main text, which
we restate below for readability. As discussed in the main text,
these theorems imply that almost all $4$-qubit states are both in $MES_4$
and isolated. Moreover, they allow the characterization of the subset of LOCC-convertible states in the MES.

As we have seen before, for an LOCC transformation from a generic $4$-qubit
state $\ket{\Psi_1} = g \ket{\Psi}$ to another $\ket{\Psi_2}= h
\ket{\Psi}$ (here $|\Psi\rangle$ is any of the seed states given in
Eq.\ (6) in the main text) to be possible, it must hold that \bea
\label{AppEqSep} \sum_k p_k \sigma_k^{\otimes4} H
\sigma_k^{\otimes4}= G,\eea where $H=h^\dagger h \equiv \bigotimes
H_i$ and $G=g^\dagger g \equiv \bigotimes G_i$ for each party
$i=1,2,3,4$. Each of these positive operators, e.\ g.\
$G_i=g_i^\dagger g_i$, is chosen such that $\tr(G_i)=1$ and,
therefore, $G_i=1/2\one + \sum_k g_k^{(i)} \sigma_k$, with
$\g^{(i)}=(g^{(i)}_1,g^{(i)}_2,g^{(i)}_3)\in \R^3$ fulfilling $0\leq
|\g^{(i)}| < 1/2$. As mentioned before, Eq. (\ref{AppEqSep}) implies
that \bea \label{AppEQ_Tensor}{\cal E}_{4}(H)={\cal E}(H_1)\otimes
{\cal E}(H_2)\otimes {\cal E}(H_3)\otimes {\cal E}(H_4). \eea As in the main text, ${\cal E}_4$ is the completely positive map given
in Eq.\ (\ref{AppEqSep}) and ${\cal E}(\rho)=\sum_{k=0}^3 p_k
\sigma_k \rho \sigma_k$. Similarly, we will use the notation ${\cal
E}_{l}(\rho)=\sum_k p_k S_k^\dagger \rho S_k$, where
$S_k=\sigma_k^{\otimes l}$ is acting on $l$ systems and $\rho$
describes $l$ systems. Note that Eq.\ (\ref{AppEQ_Tensor}) only
depends on $\ket{\Psi_2}=h\ket{\Psi}$ and is independent of
$\ket{\Psi_1}=g\ket{\Psi}$. As in the main text, we use the
notation $h^i_w\in span\{\one,W\}$, where $W\in \{X,Y,Z\}$ such that
$H_i=(h^i_w)^\dagger h^i_w=1/2\one+h^{(i)}_w W$.

\begin{appth} \label{AppThgen4} A generic state $h \ket{\Psi}$ is reachable via LOCC from some other state iff either \bi \item[(i)] $h=h^1\otimes \one^{\otimes 3}$ with $h^1\not \propto \one$ arbitrary or
\item[(ii)] $h= h^1\otimes h_w^2 \otimes h_w^3 \otimes h_w^4$, for $w\in\{x,y,z\}$ where $h^1\neq h_w^1$. \ei \end{appth}

\begin{proof}
Throughout the proof, we use lowercase indices $i\in\{1,2,3\}$ for
the coordinates of the SLOCC parameters and uppercase indices
$I\in\{1,2,3,4\}$ for the parties, e.\ g.\ $h_i^{(I)}$ denotes the $i$th component of the operator $h$ acting on system $I$.

{\it Only if:} We first show that the conditions stated in the
Theorem are necessary for a state $h \ket{\Psi}$ to be obtainable
via LOCC from some other state $g\ket{\Psi}$. From Eq.\
(\ref{AppEQ_Tensor}), it follows that ${\cal E}_{2}(H_1\otimes
H_2)={\cal E}(H_1)\otimes {\cal E}(H_2)$ must hold (and similarly
for other pairs of parties). This is equivalent to \bea
\label{Eq_sep2} \h^{(1)}(\h^{(2)})^T \bigodot (N_1-N_2)=0,\eea where
$\odot$ is the Hadamard product (i.\ e.\ entry-wise matrix
multiplication), $\h^{(I)}=(h_1^{(I)},h_2^{(I)},h_3^{(I)})^T$ for
any system $I$ and $N_1=\vec{\e} \vec{\e}^T$, with
$\vec{\e}=(\e_1,\e_2,\e_3)^T$
and \bea N_2=\begin{pmatrix} \e_0 &\e_3& \e_2\\ \e_3& \e_0& \e_1\\
\e_2 & \e_1& \e_0\end{pmatrix}. \eea Here, we use the notation
$\e_0=\sum_{k=0}^3 p_k=1$ and $\e_i=p_0+p_i-(p_j+p_k)$ (where
$i,j,k$ are assumed here to be all different).

Notice that if more than one $\e^2_i$ equals $1$, then only one
$p_k$ differs from zero, which implies that the initial and the
final states are LU equivalent. Thus, we are looking for solutions
of Eq. (\ref{Eq_sep2}) where there exists at most one $i\in
\{1,2,3\}$ such that $\e_i^2=1$.

Clearly, if $\h^{(1)}(\h^{(2)})^T=0$, which is the case iff
$\h^{(1)}=0$ and/or $\h^{(2)}=0$, Eq.\ (\ref{Eq_sep2}) is satisfied.
The same reasoning extends to other pairs of parties and this leads
to the states of case (i) of the Theorem, where we need to take into
account that for one party it must hold that $\h^{(I)}\neq0$ since
we know that seed states are in the MES. In case
$\h^{(1)}(\h^{(2)})^T\neq 0$, there can be at most one $i$ such that
$h_i^{(1)} h_i^{(2)}\neq 0$, since otherwise more than one
$\eta_i^2$ must be $1$, which can be seen by looking at the diagonal
entries of Eq.\ (\ref{Eq_sep2}). Let us consider first the case in
which no $\eta_i^2=1$. Then, $\h^{(1)} = (0,h_2^{(1)},h_3^{(1)})$ and
$\h^{(2)} = (h_1^{(2)},0,0)$ up to permutations of the parties
and/or entries of the vectors. Moreover, imposing Eq.\
(\ref{Eq_sep2}) for parties (1,2) and 3 and parties (1,2) and 4 leads to
$\h^{(3)} =\h^{(4)} =0$ (otherwise $\eta_1^2=1$). Those states are
included in case (ii) of the Theorem. Finally, let us consider the
case in which $\e_i^2=1$ for exactly one $i$. Say, without loss of
generality, $\eta_1^2=1$. Then, by looking at Eq.\ (\ref{Eq_sep2})
for other pairs of parties it must hold that $\h^{(1)}$ is
arbitrary, $\h^{(2)} = (h_1^{(2)},0,0)$, $\h^{(3)} =
(h_1^{(3)},0,0)$ and $\h^{(4)} = (h_1^{(4)},0,0)$ up to
permutations, which corresponds to the states of case (ii) of the
Theorem. However, it remains to show that $\h^{(1)}$ is not
completely arbitrary as the case $\h^{(1)} = (h_1^{(1)},0,0)$ must
be excluded. To see this, notice that ${\cal E}(H_I)=G_I$ for all
parties $I$ and $\eta_1^2=1$ imposes that $|h_1^{(I)}|=|g_1^{(I)}|$.
This, together with the fact that $|\g^{(I)}|$ cannot decrease under
deterministic LOCC transformations, which was proven in the main text, imposes
that $h_w^1\otimes h_w^2 \otimes h_w^3 \otimes h_w^4|\Psi\rangle$
cannot be reached from any other LU-inequivalent state.

 %%%%%%%%
{\it If:} Let us now show that the states given in the Theorem can
indeed be reached via LOCC. Let us first treat the states belonging to case (ii)
where, without loss of generality, we choose $w=x$. Consider
$\ket{\Psi_1}=g^1_x\otimes h^2_x \otimes h^3_x \otimes h^4_x
\ket{\Psi}$ for some $g^1_x$, which is specified below, and the two-outcome POVM $\{M_1,M_2\}$
with
\begin{align}
M_1&=\frac{1}{\sqrt{2}}h^1 (g^1_x)^{-1} \otimes \one^{\otimes
3},\nonumber\\
M_2&=\frac{1}{\sqrt{2}} h^1 X (g^1_x)^{-1} \otimes X^{\otimes
3}.\label{AppPOVM1}
\end{align}
Since $1/2 (X H_1 X + H_1)= 1/2 \one + h^{(1)}_1 X$, $\{M_1,M_2\}$
is a valid POVM whenever $h^{(1)}_1=g^{(1)}_1$. Moreover, the POVM
can be implemented by LOCC: party 1 implements the POVM $\{h^1
(g^1_x)^{-1}/\sqrt{2},h^1 X(g^1_x)^{-1}/\sqrt{2}\}$, and in case of
the second outcome the other parties implement the LU $X$. Since
$[h_x^I,X]=0$ and $X^{\otimes4}|\Psi\rangle=|\Psi\rangle$ for any
seed state $|\Psi\rangle$, both branches of the protocol lead to the
same outcome. Thus, any state $h^1\otimes h_w^2 \otimes h_w^3
\otimes h_w^4|\Psi\rangle$ with arbitrary $h^{(1)}\not\propto h_x^1$
can be obtained by LOCC from $\ket{\Psi_1}$ if $h^{(1)}_1=g^{(1)}_1$
holds.

Let us now consider the case of the states of case (i), where
$h^1$ is arbitrary, including, in particular, the case where $h^1=h_x^1$
and all the other $h^i$ are proportional to the identity. This state
can be obtained form the seed state $\ket{\Psi}$ via the POVM
$\{M_i\}_{i=0}^3$, with \bea M_i=\frac{1}{\sqrt{2}} h^1 \sigma_i \otimes
\sigma_i^{\otimes 3}. \eea Arguing as above, it is clearly seen that
this POVM can also be implemented by LOCC.
\end{proof}

\begin{appth} \label{ApptheoremISO}
A generic state $g \ket{\Psi}$ is is convertible via LOCC to some other state iff (up to
permutations) $g=g^1\otimes g_w^2 \otimes g_w^3 \otimes g_w^4$ with
$w\in \{x,y,z\}$ and $g^1$ arbitrary.
\end{appth}

\begin{proof}

{\it If:} Let us choose without loss of generality $w=x$, since the
same argument applies to any other choice. Following the same
methods as in the proof of Theorem \ref{AppThgen4}, it can be
readily seen that $\ket{\Psi_1}=g \ket{\Psi}$ with $g=g^1\otimes
g_x^2 \otimes g_x^3 \otimes g_x^4$ can be transformed into some (LU-inequivalent)
state $\ket{\Psi_2}=h \ket{\Psi}$ with $h=h^1\otimes g_x^2 \otimes
g_x^3 \otimes g_x^4$. For that, choose a similar POVM $\{M_1,M_2\}$
to that of Eq. (\ref{AppPOVM1}) where now

\begin{align}
M_1&=\sqrt{p}h^1 (g^1)^{-1} \otimes \one^{\otimes 3},\nonumber\\
M_2&=\sqrt{1-p} h^1 X (g^1)^{-1} \otimes X^{\otimes 3}
\end{align}
with $1>p>0$ arbitrary. This is indeed a POVM if $(1-p) X H_1 X +
pH_1= G_1$, which amounts to $h_1^{(1)}=g_1^{(1)}$ and
$(2p-1)h_{2,3}^{(1)}=g_{2,3}^{(1)}$. Since $0\leq
|\g^{(i)}|,|\h^{(i)}| < 1/2$ is the only condition for the states,
given any state $\ket{\Psi_1}$ of the above form, there is always
some value of $p$ so that it can be transformed by LOCC to a non-LU-equivalent state $\ket{\Psi_2}$ of the above form by increasing (by
the same proportion) the parameters $g_{2}^{(1)}$ and $g_{3}^{(1)}$.

Similarly, one can show that if $\ket{\Psi_1}=g \ket{\Psi}$ with
$g=g_1\otimes \one^{\otimes 3}$, then the state can be transformed by LOCC into another state. This is any state of the form $\ket{\Psi_2}=h^1\otimes
\one^{\otimes 3} \ket{\Psi}$ with $h^1$ such that there exists
probabilities $p_k$ and therefore values of $\e_k$ such that
$G_1={\cal E}(H_1)=1/2 \one +\sum_i h^{(1)}_i \e_i \sigma_i$. The
corresponding POVM is $\{M_i\}_{i=1}^4$ with \bea M_i=\sqrt{p_i} h^1
\sigma_i (g^1)^{-1}\otimes (\sigma_i)^{\otimes 3},\eea which
completes the sufficient part of the proof.

{\it Only if:}
Due to Theorem \ref{AppThgen4}, we know that the only states which can be reached via LOCC are of the form $\ket{\Psi_2}= h \ket{\Psi}$ with
either $h=h^1\otimes \one^{\otimes 3}$ with $h^1\not \propto \one$ arbitrary (case (i)) or $h= h^1\otimes h_w^2 \otimes h_w^3 \otimes h_w^4$, for $w\in\{x,y,z\}$ where $h^1\neq h_w^1$ (case (ii)). Thus, any LOCC-convertible state $\ket{\Psi_1}=g \ket{\Psi}$ can only be transformed into one of those states. It is easy to see now that $\ket{\Psi_1}$ must be of the form given in Theorem \ref{ApptheoremISO}, since Eq. (\ref{AppEqSep}) implies that ${\cal E}_1(H_i)=1/2\one+ \sum_k \e_k h_k^{(i)} \sigma_k=G_i$. Thus, a component of $G_i$ can only be non vanishing if the corresponding component of $H_i$ is non vanishing, which proves the statement.\end{proof}

\end{document}